\documentclass[times,twocolumn,final]{IEEEtran}

\usepackage{framed,multirow}
\usepackage[numbers]{natbib}
\usepackage{amssymb}
\usepackage{amsmath}
\usepackage{amsthm}
\usepackage{graphicx}
\usepackage{amsmath}
\usepackage{amssymb}
\usepackage{booktabs}
\usepackage{mathtools}
\usepackage[linesnumbered,ruled,vlined]{algorithm2e}
\usepackage{latexsym}
\usepackage{graphicx}
\usepackage{multirow}
\usepackage{adjustbox}
\usepackage{wrapfig}
\usepackage{makecell}
\usepackage{subcaption}
\usepackage{xcolor}
\usepackage{caption}
\usepackage{graphicx}
\usepackage{url}
\usepackage{xcolor}
\usepackage{hyperref}
\usepackage{color, colortbl}
\definecolor{Gray}{gray}{0.9}
\definecolor{LightCyan}{rgb}{0.88,1,1}
\usepackage{comment}

\usepackage[capitalize]{cleveref}
\crefname{section}{Sec.}{Secs.}
\Crefname{section}{Section}{Sections}
\Crefname{table}{Table}{Tables}
\crefname{table}{Tab.}{Tabs.}
\newtheorem{prop}{Proposition}

\begin{document}

\title{UNICORN: Ultrasound Nakagami Imaging via Score Matching and Adaptation for Assessing Hepatic Steatosis}                   

\author{Kwanyoung Kim$^{*1}$, \;  Jaa-Yeon  Lee$^{*2}$, \\ \; Youngjun Ko$^{3}$,  GunWoo Lee$^{3}$, \;
Jong Chul Ye$^{\dagger 1}$ \\
\thanks{$^{*}$: Co-first authors with equal contribution. $^{\dagger}$: Corresponding authors. \\
$^{1}$Department of AI Convergence, GIST, $^{2}$Kim Jae Chul Graduate School of AI, KAIST, $^{3}$Samsung Medison, Email: jong.ye@kaist.ac.kr.}}

\maketitle

\begin{abstract}
Ultrasound imaging is an essential first-line tool for assessing hepatic steatosis. While conventional B-mode ultrasound imaging has limitations in providing detailed tissue characterization, ultrasound Nakagami imaging holds promise for visualizing and quantifying tissue scattering in backscattered signals, with potential applications in fat fraction analysis. However, existing methods for Nakagami imaging struggle with optimal window size selection and suffer from estimator instability, leading to degraded image resolution.
To address these challenges, we propose a novel method called UNICORN (\textbf{U}ltrasound \textbf{N}akagami \textbf{I}maging via S\textbf{cor}e Matching and Adaptatio\textbf{n}), which offers an accurate, closed-form estimator for Nakagami parameter estimation based on the score function of the ultrasound envelope signal. Unlike methods that visualize only specific regions of interest (ROI) and estimate parameters within fixed window sizes, our approach provides comprehensive parameter mapping by providing a pixel-by-pixel estimator, resulting in high-resolution imaging. We demonstrated that our proposed estimator effectively assesses hepatic steatosis and provides visual distinction in the backscattered statistics associated with this condition. Through extensive experiments using real envelope data from patient, we validated that UNICORN enables clinical detection of hepatic steatosis and exhibits robustness and generalizability.
\end{abstract}

\section{Introduction}

Hepatic steatosis, characterized by the accumulation of excessive fat within liver cells forming fatty vacuoles~\cite{kawano2013mechanisms}, can advance to more severe conditions such as nonalcoholic steatohepatitis, fibrosis, cirrhosis, and hepatocellular carcinoma~\cite{adams2005recent,targher2010risk,zoller2016nonalcoholic}. Nonalcoholic fatty liver disease (NAFLD), primarily driven by hepatic steatosis, has emerged as the foremost cause of chronic liver disease~\cite{mishra2012epidemiology}. Consequently, accurate diagnosis and assessment of hepatic steatosis are crucial to preventing the progression of liver diseases, thereby enhancing opportunities for curative treatment and improving patient survival.

Liver biopsy is regarded as the gold standard for diagnosing hepatic steatosis~\cite{kleiner2005design,bravo2001liver}. However, this method has notable limitations, such as sampling error and invasive side effects~\cite{ratziu2005sampling,sumida2014limitations}. Consequently, noninvasive imaging techniques have emerged as valuable alternatives for the quantitative diagnosis of hepatic steatosis.~\cite{ma2009imaging} Among these, magnetic resonance imaging-proton density fat fraction (MRI-PDFF) stands out as a reliable diagnostic tool. Nevertheless, MRI-PDFF is not used for routine clinical assessments due to its high cost and limited availability~\cite{reeder2012proton}.
Ultrasound B-mode imaging presents an essential alternative for liver steatosis by providing diagnostic clues with properties that are safe, real-time, noninvasive, and cost-effective~\cite{schwenzer2009non,thijssen2008computer}. However, the assessment of hepatic steatosis using B-mode imaging is prone to inter-clinician variability and is limited in the tissue information. 

To overcome this limitation, quantitative ultrasound (QUS) methods have emerged to analyze statistical features of ultrasound radiofrequency (RF) echo signals in the liver, including tissue backscatter~\cite{oelze2016review} and ultrasonic wave attenuation~\cite{bigelow2005estimation}. QUS enables finer distinctions among tissue types beyond the capabilities of B-mode imaging alone. The backscattered ultrasound signals, specifically the envelope statistics, encapsulate the characteristics of scatterers within tissue, such as their shape, size, density, and other properties~\cite{bamber1981acoustic,insana1990describing}. Consequently, QUS methods utilizing these signals have been explored to visualize the scatter properties of ultrasound within tissue~\cite{tsui2007imaging}.

Initially, the Rayleigh distribution was employed to model backscattered signal data. However, scatterers in tissue exhibit varied scattering patterns, ranging from pre-Rayleigh to post-Rayleigh distributions~\cite{shankar1995model,dutt1994ultrasound}. To better account for these variations, the Nakagami distribution has been studied as a more general statistical model.
The Nakagami parameter, estimated from backscattered echoes, can characterize various backscattering distributions in medical ultrasound~\cite{zhang2012feasibility}. This parameter has been successfully used in ultrasound parametric imaging to assess hepatic steatosis, demonstrating a significant correlation between tissue scattering imaging results and MRI-PDFF measurements~\cite{ho2013relationship,lin2015noninvasive,wan2015effects}. 

Previous Nakagami imaging studies in ultrasound commonly use moment-based and maximum-likelihood estimators (MLE) for mapping Nakagami distribution parameters. In the moment-based approach, Nakagami parameters are computed within a sliding local window, which is typically optimized to be three times the transducer pulse length for optimal performance~\cite{tsui2007imaging}. Conversely, MLE-based Nakagami imaging provides more consistent results with smaller variances~\cite{cheng2001maximum}. Window-modulated compounding (WMC) Nakagami imaging utilizes a momentum estimator with varying local window sizes to enhance image smoothness~\cite{tsui2014window}. However, both moment-based and MLE rely on the sliding window technique, which involves a trade-off between image resolution and estimator stability~\cite{larrue2011nakagami}. Larger window sizes improve image smoothness but reduce resolution, while smaller window sizes offer finer resolution at the expense of stability. Therefore, selecting the optimal window size is crucial and depends heavily on the specific hyperparameters used.
%

To address these issues, here we introduce a novel framework for \textbf{U}ltrasound \textbf{N}akagami \textbf{i}maging via S\textbf{cor}e Matching and Adapatio\textbf{n}, termed UNICORN. 
Inspired by the success of the self-supervised denoising approach that utilizes the score function, i.e. the gradient of log-likelihood \cite{kim2021noise2score,kim2022noise},
this framework offers a closed-form solution for mapping Nakagami parameters in terms of the {\em score function} of the measurement. Specifically, instead of using momentum and MLE, our estimator calculates the posterior mean and provides an optimal solution from a Bayesian perspective. By integrating the score function of RF envelope signals, UNICORN directly computes Nakagami images, eliminating the need for the sliding window technique. Consequently, our proposed technique preserves ultrasound imaging resolution in Nakagami imaging while ensuring stability. Our contributions can be summarized as:

\begin{itemize}
\item To the best of our knowledge, we present the first closed-form solution for estimating the Nakagami parameter using the score function of RF envelope data, which addresses limitations of conventional methods.
\item Through extensive experiments including simulation and real RF envelope data, our proposed method, called UNICORN, demonstrates its superiority in estimation performance and the distinction between normal and fatty liver. 
\item Extensive analysis on a real-patient dataset validated the clinical effectiveness of our proposed Nakagami imaging method in diagnosing hepatic steatosis.
\end{itemize}

\begin{figure*}[t]
  \begin{center}
    \includegraphics[width=1\textwidth]{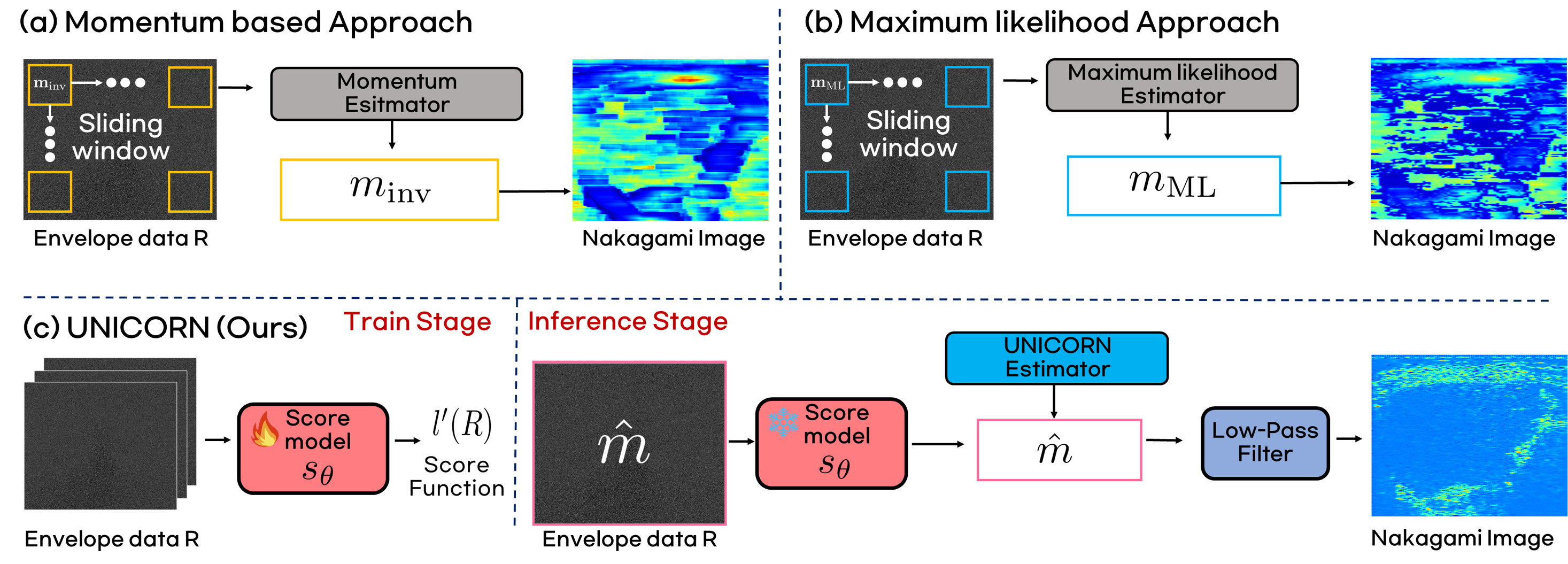}
  \end{center}
  \caption{
Nakagami imaging using conventional methods and our UNICORN framework.
(a) Momentum-based approach uses a sliding window for Nakagami parameter calculation.
(b) Maximum likelihood method obtains Nakagami image through ML estimation.
(c) UNICORN consists of two stages: training a score model to learn RF envelope score function, and inference step estimates Nakagami image in terms of score function.}
  \label{fig:main}
\end{figure*}

\section{Related Work}
\subsection{Quantitative Assessment for Hepatic Steatosis}

Recently, a number of quantitative ultrasound (QUS) techniques have been developed and validated for diagnosing hepatic steatosis by exploiting the intrinsic acoustic properties of tissue to characterize its microstructure \cite{li2024correlation, rou2024assessment, ghoshal2012ex, zhou2019hepatic}. These approaches typically measure the attenuation coefficient, backscatter coefficient, and sound‑speed, all of which are sensitive to hepatic fat content~\cite{ferraioli2021quantification, jang2023non}. In fatty liver, increased scattering leads to a higher backscatter coefficient compared with normal liver tissue \cite{park2022quantitative}. To detect early steatosis, both parametric backscatter analyses\cite{sato2021fatty} and non‑parametric statistical methods~\cite{lin2024clinical, chan2021ultrasound} have been proposed. Parametric backscatter models are often expressed with generalized probability‑density functions that encompass backscattering distributions such as the homodyned‑K \cite{dutt1994ultrasound} and Nakagami \cite{tsui2016application, wan2015effects} models. Numerous experimental studies have reported a positive correlation between backscatter metrics and steatosis grade \cite{zagzebski1993quantitative, bamber1981acoustic, zhou2018hepatic}. Among the many QUS techniques, this work concentrates on Nakagami‑based ultrasound imaging for the assessment of hepatic steatosis.

\subsection{Ultrasound Nakagami Imaging}

The Nakagami distribution is widely used in speckle analysis to characterize the statistical properties of tissue scattering in ultrasound imaging~\cite{ho2012using, zhang2012feasibility}. It offers a simple moments‑based method for estimating the shape parameter $m$
, which is easy to implement and yields low bias. However, MLE estimators provides a lower‑variance estimate of $m$
~\cite{destrempes2023review}. To reduce the error of MLE‑based $m$
 estimates, several approximations have been proposed, including the Tolparev‑Polyakov, Lorenz, Greenwood‑Durand, and recursive Bowman estimators~\cite{kolar2004estimator}. The WMC technique further improves the smoothness of parametric maps by fusing results from different window sizes, thereby balancing spatial resolution and variance~\cite{tsui2014window, zhou2018three}. In addition, coarse‑to‑fine approaches that employ pyramidal layers effectively preserve details in heterogeneous tissues~\cite{han2017nakagami}.

Ultrasound Nakagami imaging has been applied to a variety of tissue‑characterization tasks,  including liver fibrosis~\cite{ho2012using, tsui2016application, fang2020ultrasound}, thermal lesions~\cite{rangraz2014nakagami, zhang2012feasibility}, and breast masses~\cite{shankar2001classification, tsui2010ultrasonic}.  The technique has shown particular promise for fatty‑liver assessment through backscatter‑statistics analysiss~\cite{ho2013relationship, wan2015effects}. One notable study used the Nakagami shape parameter to quantify liver‑tissue echo‑amplitude statistics, distinguishing pre‑Rayleigh, Rayleigh, and post‑Rayleigh regimes~\cite{tsui2016acoustic}. This work highlights the potential of Nakagami imaging for a more accurate evaluation of hepatic steatosis.

Nevertheless, most previous studies relied on existing estimators (moments‑based or MLE). In contrast, we introduce a novel estimator based on the score function and posterior mean. To the best of our knowledge, this is the first such estimator proposed, and we validate its performance using real-patient data.

\section{Background}
\subsection{Nakagami Distribution}
The probability density function of the ultrasound backscattered envelope $R$ under the Nakagami statistical model can be described as follows:
\begin{align}
    p_{R}(r|m,\Omega) = \frac{2}{\Gamma(m)}\left(\frac{m}{\Omega}\right)^{m}r^{2m-1}\exp\left(-\frac{m}{\Omega}r^2\right)\mathcal{U}(r) 
    \label{pd:naka}
\end{align}
where $r$ is a value of the random variable $R$,
$\Gamma(\cdot)$ and $\mathcal{U}(\cdot)$ represent the Gamma function and the unit step function, respectively. $\Omega$ denotes the scale parameter,
while $m$ is the Nakagami parameter or shape parameter that determines the statistical distribution of the ultrasound backscattered envelope~\cite{nakagami1960m}. Specifically, a Nakagami parameter ranging from 0 to 1 indicates a transition from a pre-Rayleigh to a Rayleigh distribution, while a parameter larger than 1 implies that the statistics of the backscattered signal conform to post-Rayleigh distributions.

\subsection{Momentum Estimators}
The momentum estimators are classical methods used for estimating parameters in statistical methods. Nakagami~\cite{nakagami1960m} provides the momentum estimator for the Nakagami parameter $m$ as follows:
\begin{align}
    m_{\text{inv}} = \frac{[\mathbb{E}(R^2)]^2}{\mathbb{E}[R^2 - \mathbb{E}(R^2)]^2}, \quad \hat \Omega = \mathbb{E}[R^{2}], \label{m_estimator}
\end{align}
where $m_{\text{inv}}$ denotes the esimated Nakagmi paramter by the momentum estimator.
The second moment of the Nakagami distribution is $\Omega$. Therefore, $\hat \Omega$ is an unbiased estimator for $\Omega$.
By employing the momentum estimators, ultrasound Nakagami imaging can be obtained by adopting the sliding window algorithm~\cite{tsui2007imaging,tsui2010ultrasonic}. In other words, the local Nakagami parameter is estimated within a square window of a specific size and assigned to the new pixel located at the center of the window. This process is repeated until the window covers the entire envelope image by sliding across it, as illustrated in \cref{fig:main} (a).

\subsection{Maximum Likelihood Estimator}
In MLE approach~\cite{cheng2001maximum}, the Nakagami parameter $m_{\text{ML}}$ is determined by maximizing the likelihood function. Let $R_1$, $\cdots$, $R_N$ be random variables which are i.i.d according to \cref{pd:naka}. The log-likelihood function of the independent multivariate Nakagami distribution is given by:
\begin{align}
   \mathcal{L}(m,\Omega|r) =-N \log(\Gamma(m)) + N m \log(m) - N m\log(\Omega) \\ \nonumber
    + (2m-1) \sum^{N}_{i=1} \log(r_i) -\frac{m}{\Omega} \sum^{N}_{i=1} r_i^2
\end{align}
where $\{r_i, i = 1,\cdots,N \}$ are the samples of $\{R_i, i = 1,\cdots,N \}$.
By differentiating the log-likelihood with respect to $m$ and setting it equal to zero, we obtain the following expression:
\begin{align}
\log(m) - \psi(m) = \log \left(\frac{1}{N} \sum^{N}_{i=1} r_i^2\right) - \frac{1}{N} \sum^{N}_{i=1} \log(r_i^2) \label{mle}
\end{align}
where $\psi(m)$ is the digamma function, which is the derivative of the logarithm of the Gamma function. To solve nonlinear ~\cref{mle}, the Taylor approximation is adopted, $\psi(m) \approx \log(m) -(1/2m)$ :
\begin{align}
    m_{\text{ML}} = \frac{1 
    }{2 (\log \left(\frac{1}{N} \sum^{N}_{i=1} r_i^2\right) - \frac{1}{N} \sum^{N}_{i=1} \log(r_i^2)))} \label{mle_estimator}
\end{align}
where $m_{\text{ML}}$ denotes the esimated Nakagmi paramter using ML estimator. Similar to the momentum estimator,
the ML estimator also adopts a sliding window approach to calculate the locally estimated Nakagami parameter as shown in \cref{fig:main} (b).

\subsection{Window modulated Compound Estimator}

In WMC estimator~\cite{tsui2014window}, the Nakagami parameter $m_{com}$ is calculated by averaging $K$ different Nakagami parameters $\{m_i, i = 1,\cdots,K \}$, each obtained from a momentum estimator with a different window size \cref{m_estimator}:
\begin{align}
    m_{com}=\frac{1}{K}\sum^{K}_{i=1}m_i
\end{align}
where $K$ denotes the number of different window size, $m_{com}$ denotes the Nakagami parameter estimated by WMC estimator.

\section{Methods}

Recent works of Noise2Score \cite{kim2021noise2score,kim2022noise} provided a highly
efficient closed form formula for self-supervised image denoising using Tweedies' formula that utilizes the score function, i.e. the gradient
of loglikelihood of the measurement.
Inspired by this, here we introduce a novel framework for Ultrasound Nakagami Imaging with Score Matching and Adaptation, termed UNICORN. UNICORN consists of two steps: firstly, we learn the score function of the ultrasonic envelope data via denoising score matching loss. Then, we estimate the Nakagami parameter $m$ per pixel, followed by a low-pass filter, and reconstruct the Nakagami Imaging, as illustrated in \cref{fig:main} (c).
In Section \ref{sec:estimator}, we provide details of our proposed estimator for the Nakagami parameter.
In Section \ref{sec:adapt}, we explain the adaptation with low-pass filter.
In Section \ref{sec:loss}, we introduce the loss function for learning the score function.

\subsection{Nakagami parameter Estimator with Score function} \label{sec:estimator}
Instead of maximizing the likelihood function with respect to $m$ in the MLE~\cite{cheng2001maximum}, which entails complex computation and approximation errors, we calculated the posterior mean of the Nakagami parameter using the envelope data $R$, thereby achieving Minimum Mean Square Error Estimation (MMSE). According to Bayes' rule, the joint posterior distribution can be represented as follow: 
\begin{align}
    p_{R}(m, \Omega |r) = p_{R}(r|m,\Omega)p(m)p(\Omega)/p_{R}(r) 
\end{align}
where $p(m)$, $p(\Omega)$, and $p_R(r)$ denote the marginal distribution of $m$, $\Omega$ and $r$, respectively.
The mode of the posterior distribution can be obtained by finding the maximum of $ p_{R}(m, \Omega|r)$. Specifically, by calculating the gradient of $\log p_{R}(m, \Omega|r)$ with respect to $r$ and setting it zero, the posterior estimate of the parameter $m$ should satisfy the following equality:
\begin{align}
   \nabla_r \log p_{R}(r|m, \Omega) - \nabla_r \log p_{R}(r) = 0, \label{key equation}
\end{align}
where $\nabla_r \log p_{R}(r|m, \Omega)$ and $\nabla_r \log p_{R}(r)$ are score function. 
By using this equality, we can derive the following closed-form estimator for the Nakagami parameter. 
\begin{prop}
For the given measurement model \eqref{pd:naka},
the estimate of the unknown Nakagami parameter $m$ is given by
\begin{align}
\hat m = \mathbb{E}[m|r] = \frac{\frac{1}{r} + \nabla_r \log p_{R}(r) }{\left(\frac{2}{r} - \frac{2r}{\hat \Omega}\right)}, 
\end{align}
where 
$\nabla_r \log p_{R}(r)$ is the score function of the RF envelope data $R$ and $\hat \Omega = \mathbb{E}[R^2]$ 
\label{prop}
\end{prop}
\begin{proof}
See \cref{proof1}.
\end{proof}
Notably, in the derivation of \cref{prop}, no approximations are necessary, leading to more accurate estimations compared to existing methods. We demonstrated that the posterior mean of the Nakagami parameter can be calculated using a nonlinear formula involving the score function of the measurement. To the best of our knowledge, this is the first attempt to estimate the Nakagami parameter using the score function of the measurement.

\subsection{Adaptation via Low-Pass Filter} \label{sec:adapt}
In conventional methods, the local Nakagami parameter is calculated within a predefined window size to stabilize the estimation at the expense of resolution. Instead of adopting a sliding window approach, our proposed method uses a pixel-by-pixel estimator to provide high-resolution images. To further enhance the robustness of the algorithm against outliers, we incorporate a small-size low-pass filter, such as a median or average filter, to the estimated parameters $\hat{m}$ as follows:
\begin{align}
m_{\text{UNICORN}} = \texttt{low-pass Filter} (\hat m)
\end{align}
where $m_{\text{UNICORN}}$ is the final Nakagami parameter obtained by our proposed method. 
As indicated in \cref{fig:main} (c), we obtain the final Nakagami image map using the pixel-by-pixel estimator combined with a low-pass filter. 

\subsection{Loss function for Denoising Score Matching} \label{sec:loss}

To learn the score function from the ultrasound backscattered envelope $R$, 
we employ the amortized residual DAE (AR-DAE), which is a stabilized implementation of denoising autoencoder \cite{lim2020ar} 
Specifically,  AR-DAE loss function is defined by:
\begin{align}
\underset{\Theta}{\arg \min }  ={\underset{R \sim p_{R}(r)}{\mathbb{E}}}\|u + \sigma_a s_\theta(R + \sigma_a u)\|^2
\label{loss:ar-dae}
\end{align}
where $s_{\theta}$ is the score model parameterized by $\theta$, $u \sim \mathcal{N}(0,I)$, and $\sigma_a$ $\sim \mathcal{N}(0,\delta^2)$. $\sigma_a$ is perturbed noise which gradually decreases with an annealing schedule. Minimizing Eq. (\ref{loss:ar-dae}) provides the network $s_{\theta^{\ast}}$ which can directly estimate the score function of envelope data, $s_{\theta^{\ast}} = \nabla_{R}\log p_{R}(r) = l'(r)$.
Estimating the score function of measurements via \cref{loss:ar-dae} has been demonstrated to be both direct and stable \cite{kim2021noise2score,kim2022noise}. Therefore, we adopt this approach as the initial step of our method (see \cref{fig:main} (c)).

\section{Experimental setting}
In this work, we conduct two categories of experiments to validate the effectiveness of our proposed method. Firstly, we perform simulations using synthetic Nakagami distributions on grayscale image dataset and ultrasound image dataset. Secondly, we apply our method to real ultrasound RF envelope datasets to validate the effectiveness of distinction between normal liver and fatty liver.

\subsection{Datasets}

\noindent\textbf{Synthetic Simulation Experiments} We conduct evaluations on both the MNIST and the BUSI ultrasound image datasets. For the MNIST dataset, we train our neural network to learn the score function using the training set and evaluate its performance on the test set.
For the BUSI dataset~\cite{al2020dataset}, which consists of breast ultrasound B-mode images from women, we adopt it for our ultrasound image dataset. 
For simulation experiments, we initially normalize images from the range of 0 to 1 to the range of 0.5 to 2 and set it as the ground truth $m$. Subsequently, we apply the Nakagami distribution (as described in \cref{pd:naka}) per pixel to generate synthetic measurement with $\Omega = 1$.

\noindent\textbf{In-vivo Liver Experiment} 
To validate the effectiveness of clinical usefulness, we evaluate the in-vivo liver dataset. In vivo liver experiments were conducted using data collected from Samsung Medical Center (SMC). The hospital's Institutional Review Board has approved this study (IRB No. SMC 2022-11-026). Written informed consents were obtained from all participants to utilize their data. Ultrasound radiofrequency (RF) data was acquired using the RS85 system and a convex probe (CA1-7S) from Samsung Medison. Multiple ultrasound scans of the right intercostal plane were performed by radiologists at each institution to visualize the hepatic parenchyma of the right liver, and RF data was automatically recorded for each scan. Magnetic resonance imaging proton density fat fraction (MRI-PDFF) data was obtained according to the standardized protocol recommended by the Quantitative Imaging Biomarkers Alliance (QIBA)~\cite{shukla2019quantitative}. The distribution of hepatic steatosis severity grades for patients recruited by each institution is summarized in \cref{tab:pdff}.

\subsection{Implementation Details}
To ensure a fair comparison with other methods, we compare our proposed method with conventional approaches such as the momentum estimator~\cite{tsui2007imaging}, MLE approach~\cite{cheng2001maximum}, and WMC estimator~\cite{tsui2014window}. The baseline window size is set to a side length equal to three times the pulse length of the incident ultrasound; for the WMC method the window size is varied from three to five times the pulse length.

For training the score network, we use the following batch sizes: 512 for the MNIST dataset, 16 for the BUSI dataset, and 4 for the in‑vivo dataset. 
 All datasets are trained for 50 epochs with an initial learning rate of 
 2 $\times$ 10$^{-4}$
, which is halved after 25 epochs. Optimization is performed with AdamW (weight decay=0.01). Data augmentation includes horizontal flipping, rotation, and random cropping to 
256$\times$256 pixels. The model is implemented in PyTorch and run on a singleNVIDIA GeForce 3090 GPU. For the in‑vivo experiments we employ the SMC dataset with 5‑fold cross‑validation.
Low‑pass filtering is applied as follows: a median filter (kernel size=7) for the synthetic experiments and an average filter (kernel size=7) for the in‑vivo experiments.

\begin{table}[!t]
\caption{The number of subjects across hepatic steatosis stages determined by MRI PDFF on each dataset. One frame per patient is used for statistical analysis.}
\centering
\scriptsize
\resizebox{1\linewidth}{!}{
    \begin{tabular}{lcccc}
    \toprule
    & \multicolumn{4}{c}{\textbf{MRI-PDFF}} \\
    \cmidrule(lr){2-5}
    Dataset & $<5\%$ (Normal) & $5\%$-$15\%$ (Mild) & $\geq 15\%$ (Severe) & Total \\
    \cmidrule(l){1-1} \cmidrule(l){2-2} \cmidrule(l){3-3} \cmidrule(l){4-4} \cmidrule(l){5-5}
    SMC & 224 & 170 & 55 & 449 \\
    \bottomrule
    \end{tabular}
}
\label{tab:pdff}
\end{table}

\subsection{Evaluation Metric}

\noindent\textbf{Synthetic Simulation Experiments} Since we generate synthetic measurements under the Nakagami distribution using the ground truth image, we evaluate the performance of the simulation results using the PSNR (Peak Signal-to-Noise Ratio) and RMSE (Root Mean Square Error) metrics across various methods.
Specifically, we calculate the metrics with the following equation:
\begin{align}
\text{PSNR} = 10 \log_{10} \left( \frac{{\text{MAX}_{I}^2}}{{\text{MSE}}} \right), \quad
\text{RMSE} = \sqrt{\text{MSE}}
\end{align}
where \(\text{MAX}_{I}\) is the maximum possible pixel value of the image and \(\text{MSE}\) is the Mean Square Error between the ground truth image and the estimated image.

\noindent\textbf{In-vivo Liver Experiment}\label{sec:eval}
In this experiment, we are unable to access the ground truth data for Nakagami imaging itself, preventing a quantitative evaluation against a reference. Instead of ground truth Nakagami imaging, we utilize the MRI-PDFF value with paired envelope data to evaluate our proposed method. The MRI-PDFF is considered a reliable quantification metric of liver fat fraction.
Using MRI-PDFF, we conduct various statistical analyses. First, we calculate the correlation between the estimated Nakagami parameter in the liver region and MRI-PDFF to validate clinical usefulness. Furthermore, we define steatosis stages based on MRI-PDFF extents as indicated in \cref{tab:pdff}. We then analyze the differentiation between stages using box plots and Receiver Operating Characteristic (ROC) curves with Area Under the ROC Curve (AUROC). 

In normal liver tissue, ultrasound backscatter is typically lower and more uniform. Fatty liver, however, exhibits increased backscatter due to fat deposits and altered tissue architecture, leading to more intense reflection and scattering. The Nakagami parameter, which quantifies backscatter, is therefore higher in fatty liver, aiding in its distinction from normal liver. Thus, we assess qualitative results based on the differences in Nakagami parameters between normal and fatty liver cases.



\begin{figure*}[!ht]
\centering
\includegraphics[width = 0.9\linewidth]{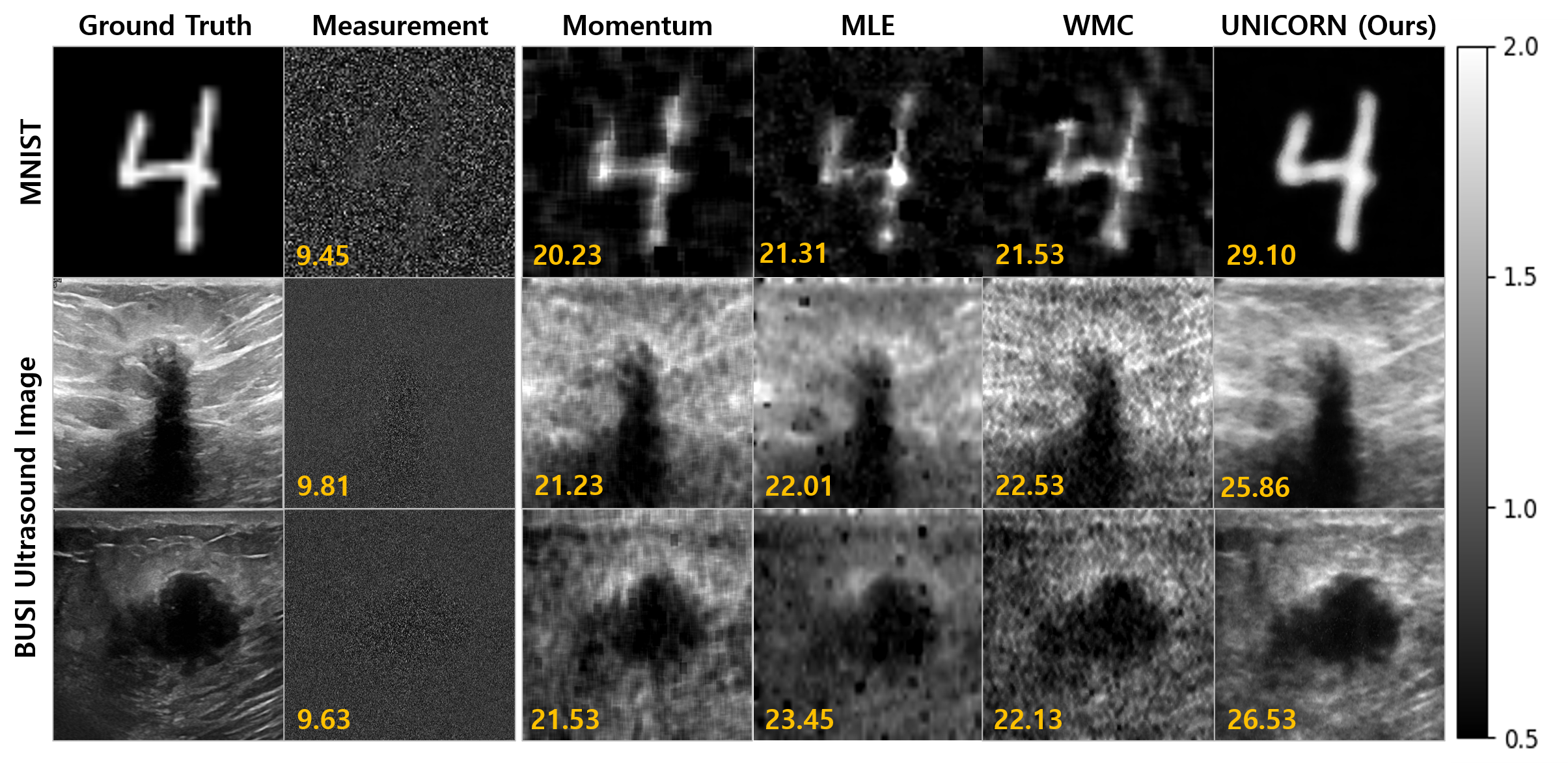}
\caption{Comparison of qualitative results across various methods: (Row 1) MNIST dataset, (Rows 2-3) BUSI Ultrasound Image dataset. Comparison against Momentum, MLE, WMC, and {UNICORN}. The yellow numbers indicate PSNR. To visualize the results, we normalize the images and convert them to grayscale.}
\label{fig}
\end{figure*}

\begin{table}[!ht]
	\centering
	\caption{Quantitative simulation results for MNIST dataset and ultrasound image BUSI dataset using various methods. The \textbf{bold} numbers indicate the best performance. WS denotes sliding window size.}
	\resizebox{1\linewidth}{!}{
		\begin{tabular}{ccccccc}
			\toprule
			Dataset & \multicolumn{3}{c}{MNIST} & \multicolumn{3}{c}{BUSI} \\ 
            \cmidrule(r){1-1} \cmidrule(r){2-4} \cmidrule(r){5-7} 
			Metric & WS & PSNR (dB) $\uparrow$ & RMSE $\downarrow$ & WS & PSNR (dB) $\uparrow$ & RMSE $\downarrow$ \\ 
			\cmidrule(r){1-1} \cmidrule(r){2-4} \cmidrule(r){5-7} 
		  Measurement  & - & 9.17 & 0.695 & - & 10.14 & 0.651 \\ 	
   		Momentum & 9 & 21.01 & 0.179 & 9 & 18.58 & 0.230 \\ 
            Momentum & 11 & 21.28 & 0.173 & 13 & 21.22 & 0.174 \\
            Momentum & 13 & 20.78 & 0.184 & 25 & 23.31 & 0.132 \\
            MLE & 11 & 20.45 & 0.191 & 25 & 22.65 & 0.148\\ 	
            WMC & 9,11,13 & 21.69 & 0.165 & 9,13,25 & 22.28 & 0.154 \\ 	
            WMC & 7,9,11,13,15 & 21.53 & 0.168 & 9,13,17,21,25 & 23.01 & 0.141 \\ 	
            \cmidrule(r){1-1} \cmidrule(r){2-4} \cmidrule(r){5-7}
			UNICORN (Ours)  & - & \textbf{28.28}(\textcolor{blue}{+6.75}) & \textbf{0.077}(\textcolor{blue}{-0.091}) & - & \textbf{25.72}(\textcolor{blue}{+2.11}) & \textbf{0.011}(\textcolor{blue}{-0.130}) \\ 			
			\bottomrule    
		\end{tabular}
	}
	\vspace{-0.5em}
	\label{tbl:main}
\end{table}

\section{Results}
\subsection{Synthetic Simulation dataset}
In \cref{tbl:main}, we provide a quantitative performance comparison using various methods on the MNIST and BUSI datasets in terms of PSNR and RMSE. For the MNIST dataset, the PSNR of the measurement method is only 9.17 dB compared to the ground truth, indicating the highly ill-posed nature of the problem. 
Momentum-based approach results were sensitive to varied window sizes. Since MLE needs higher computational cost compared to the momentum-based approach, we select the best window size from momentum-based approaches for MLE and set the step size to half of the window size. 
Our proposed method, UNICORN, does not require any window size optimization but still achieves the highest PSNR of 28.28 dB and the lowest RMSE of 0.077, surpassing all other methods. UNICORN also outperforms the existing state-of-the-art method by a significant margin of +6.75 dB in PSNR. Similarly, on the BUSI dataset, UNICORN achieves superior performance with a PSNR of 25.71 dB and an RMSE of 0.089, representing a margin of +2.7 dB in PSNR.

In \cref{fig}, we compare the qualitative results of our proposed method against existing methods. We observe that the results obtained with Momentum-based approaches and the MLE approach are too blurred out and exhibit artifacts, such as the tile effect. While the WMC approach mitigates blurring to some extent, it introduces speckle noise artifacts. In contrast, our proposed method is capable of closely estimating Nagakami parameter mapping compared to the ground truth label in MNIST and provides superior resolution without significant artifacts in BUSI dataset.


    

\begin{figure*}[!t]
\centering
\includegraphics[width = 0.85\linewidth]{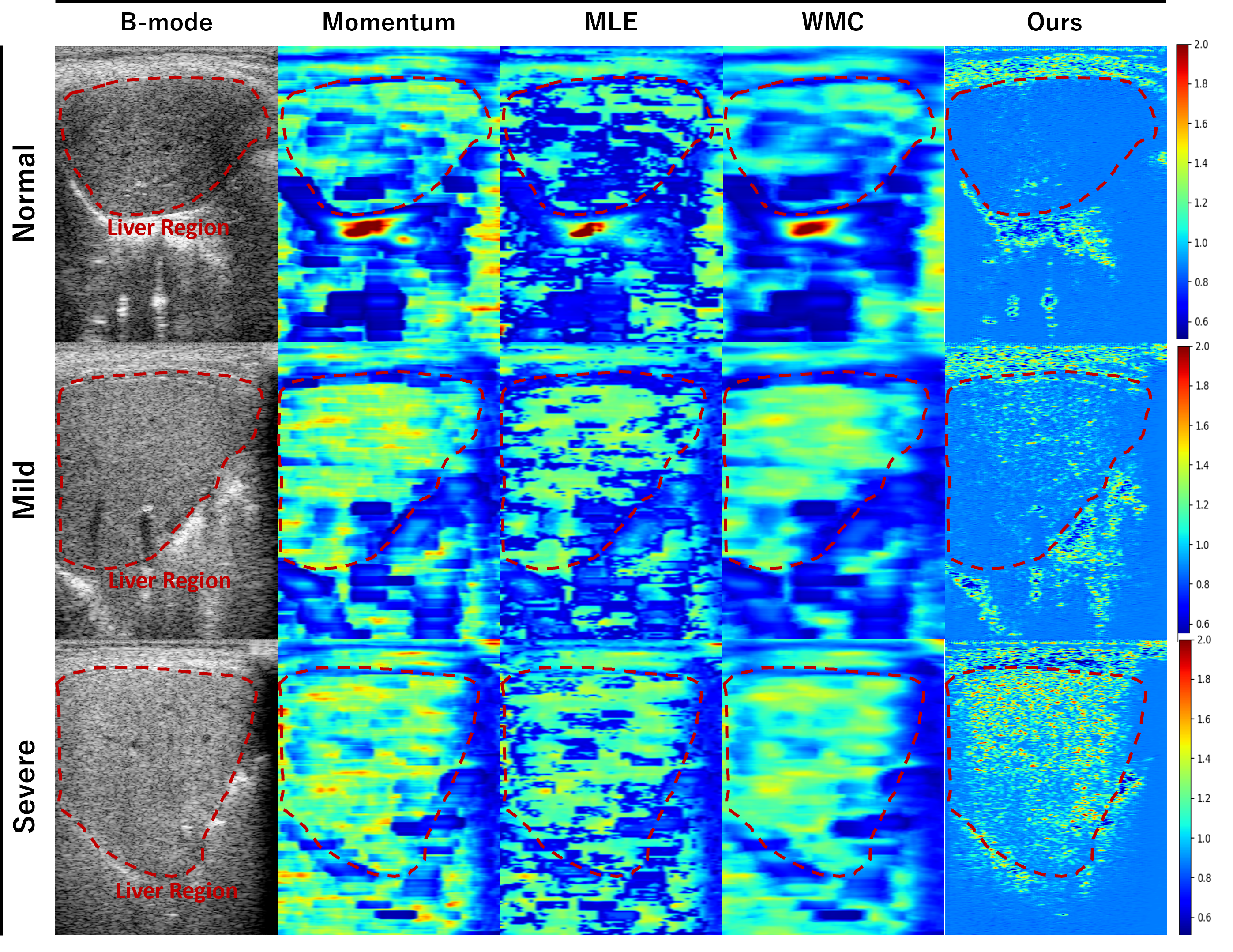}
\caption{Comparison of qualitative results obtained using different methods. Each row corresponds to Normal, Mild, and Severe liver cases, with MRI-PDFF values of 0.71, 6.34, and 22.12, respectively. The comparison includes the Momentum, MLE, WMC, and UNICORN methods. Red dashed lines delineate the liver region. Compared to other methods, our approach provides clearer visual cues for distinguishing between normal, mild, and severe fatty liver cases.}
\label{fig:asia}
\end{figure*}


\subsection{Qualitative Results on Real Clinical data}

Using the RF envelope dataset from real ultrasound imaging, we compare qualitative results with various methods as illustrated in Figure \ref{fig:asia}. The figure displays results from the SMC in‑vivo dataset for three representative cases: a normal liver, a mildly steatotic liver, and a severely steatotic liver.

Conventional estimators, namely the momentum estimator~\cite{tsui2007imaging}, the MLE estimator~\cite{cheng2001maximum}, and the WMC estimator~\cite{tsui2014window}, do produce activation values within the liver region, but the activations for mild versus severe fatty liver are hardly distinguishable, making visual discrimination ambiguous at best.

Conversely, both mild and severe fatty‑liver cases show markedly elevated Nakagami values that are confined to the liver region, with a clear gradient reflecting the degree of steatosis. This striking contrast enables an instantaneous and reliable distinction between normal and fatty liver tissue, allowing the practitioner to assess the severity of steatosis at a glance. Hence, UNICORN not only overcomes the subtlety problem of existing techniques but also sets a new benchmark for intuitive, image‑based assessment of hepatic steatosis.



\subsection{Analysis on Real Clinical data}

\noindent\textbf{Correlation Between Nakagami Parameter and MRI‑PDFF}
To further demonstrate the effectiveness of UNICORN, we calculate the Pearson correlation coefficient (PCC) between the estimated Nakagami parameter $m$ by each method and MRI-PDFFs, and we test the statistical significance of the correlations as indicated in \cref{tab:pcc}. The momentum estimator and WMC estimator yield moderate PCC values ranging from 0.53 to 0.544, whereas the MLE method achieve a slightly higher PCC of 0.65. In contrast, UNICORN produce the strongest correlation, with a PCC of 0.77, and the associated $p$‑values are significant on dataset.

\cref{fig:scatterplot} illustrates the scatter plots of $m$
 versus MRI‑PDFF for the baseline methods and UNICORN.  The conventional approaches tend to saturate once $m$ exceeds 5, indicating limited scalability and a markedly non‑linear relationship with MRI‑PDFF. UNICORN, however, maintains an approximately linear trend across the entire range of $m$ , preserving scalability and yielding a monotonic increase in MRI‑PDFF for all subjects. This linearity highlight UNICORN’s superior ability to model the relationship between the Nakagami parameter and liver fat fraction, representing a clear improvement over existing techniques.

\begin{figure*}[!ht]
\centering
\includegraphics[width = 0.9\linewidth]{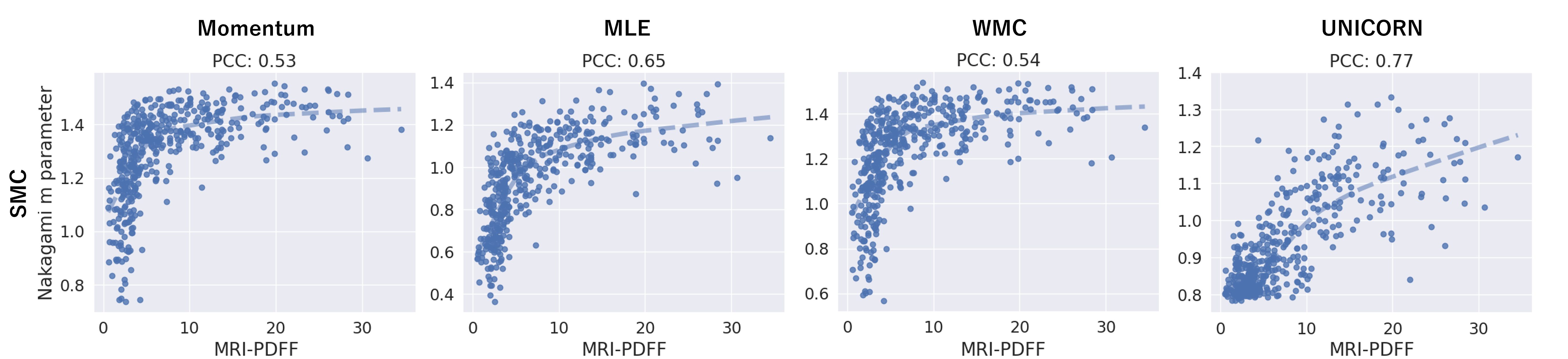}
\caption{Scatter plots illustrating the relationship between MRI-PDFF and estimated m parameters across various methods and datasets: Columns 1-4 represent Momentum, MLE, WMC, and UNICORN, respectively; PCC denotes the Pearson Correlation Coefficient between MRI-PDFF and the estimated $m$ parameter. The trend line follows a locally weighted linear regression model.}
\label{fig:scatterplot}
\end{figure*}

\begin{table}[!ht]
    \caption{Comparison of Pearson Correlation Coefficient (PCC) between MRI-PDFF and estimated Nakagami $m$ parameters and Statistical Significance (SS) across various methods} 
    \centering
    \scriptsize
    \resizebox{0.6\linewidth}{!}{
    \begin{tabular}{lcc}
          \toprule
        \multirow{2}{*}{Method}  & \multicolumn{2}{c}{SMC}   \\    
               \cmidrule(l){2-3} 
            & PCC(r) & SS  \\    
          \cmidrule(l){1-1} \cmidrule(l){2-2}  \cmidrule(l){3-3}  
         Momentum & 0.53 & ****  \\
         MLE  &   0.65 & ****  \\
         WMC  &   0.54 & **** \\
         \cmidrule(r){1-1} \cmidrule(r){2-3} 
         UNICORN(Ours) & \bf{0.77} & **** \\
    \bottomrule 
    \end{tabular}
    }
    \label{tab:pcc}
\end{table}

\begin{figure}[!t]
\centering

\includegraphics[width = 0.85\linewidth]{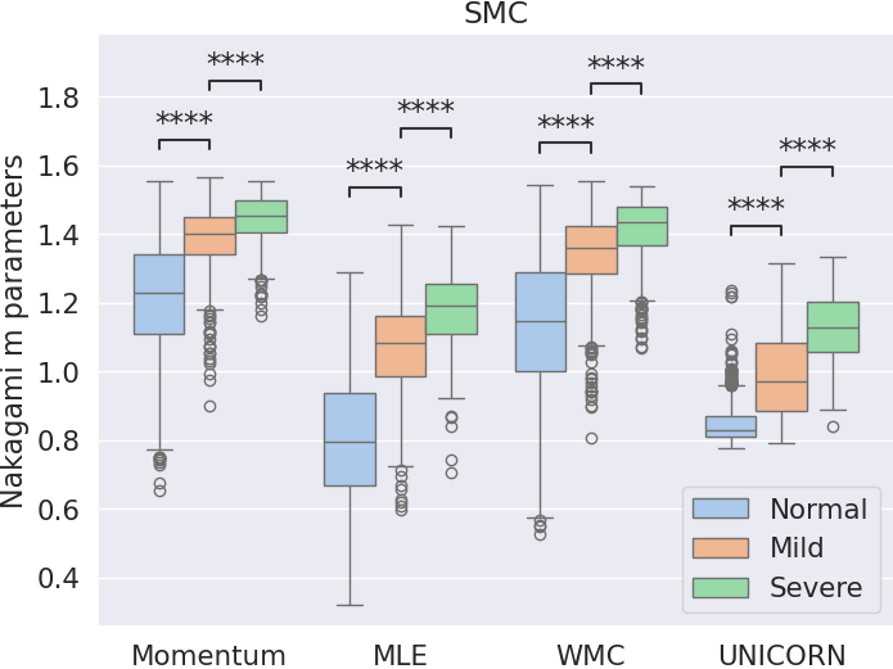}
\caption{Box plots illustrating the relationship between MRI-PDFF and estimated m parameters across different methods at each stage: Methods include Momentum, MLE, WMC, and UNICORN; stages are defined as Normal (MRI-PDFF $<$ 5$\%$), Mild (5$\%$ $\leq$ MRI-PDFF $<$ 15$\%$), and Severe (MRI-PDFF $>$ 15$\%$). The $p$-values between neighboring stages are indicated within the box plots.}
\label{fig:boxplot}
\end{figure}

\begin{figure*}[!t]
\centering
\includegraphics[width = 0.8\linewidth]{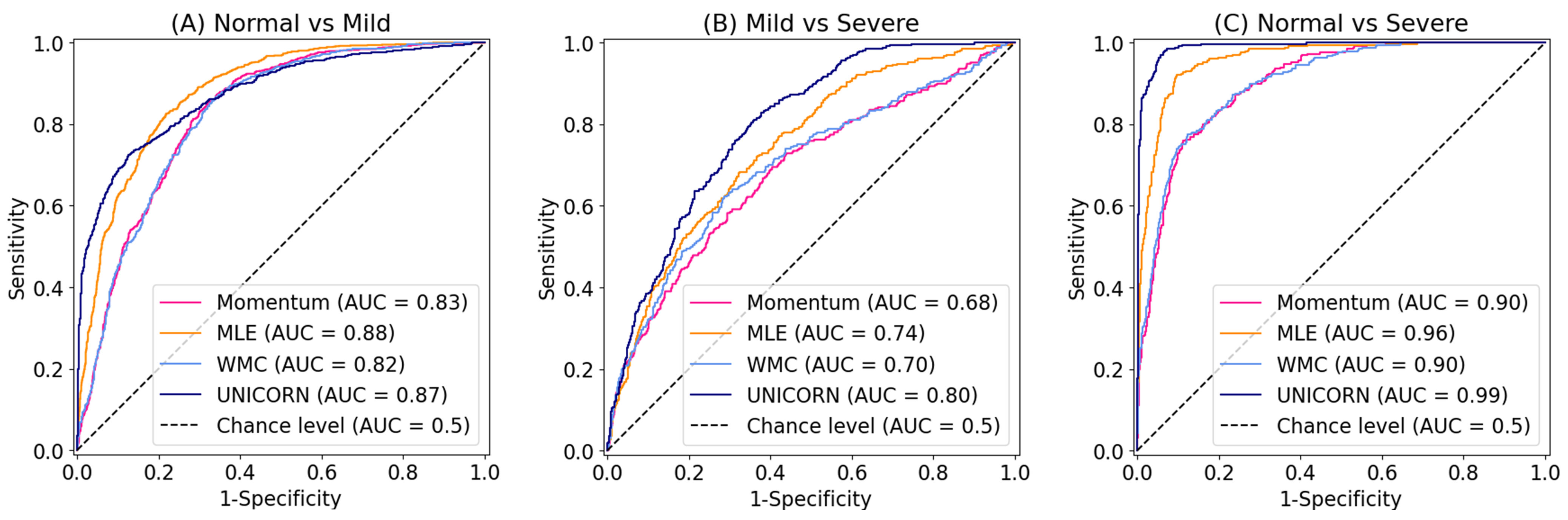}
\caption{ROC curves for estimated m parameters in the differentiation of hepatic steatosis stages using overall datasets across different methods: Momentum, MLE, WMC, and UNICORN. Stages are defined as Normal (MRI-PDFF $<$ 5$\%$), Mild (5$\%$ $\leq$ MRI-PDFF $<$ 15$\%$), and Severe (MRI-PDFF $>$ 15$\%$). (A) Normal vs Mild, (B) Mild vs Severe, and (c) Normal vs Severe.}
\label{fig:auroc}
\end{figure*}

\noindent\textbf{Investigation on Different Fact Fraction Stages}
To further validate UNICORN’s effectiveness, we compare the estimated Nakagami shape parameter $m$ across three fat‑fraction stages as depicted in \cref{fig:boxplot}. We categorize the stages into three groups: normal (MRI-PDFF $<$ 5$\%$), mild (5$\%$ $\leq$ MRI-PDFF $<$ 15$\%$), and severe (MRI-PDFF $>$ 15$\%$). Each box plot displays the distribution of $m$ within a stage, with Welch’s t-test p-values quantifying separability between adjacent stages.
Fig.~\ref{fig:boxplot} reveals that baseline methods yield sparse, high‑variance distributions, while UNICORN produces a stable, low‑variance distribution, particularly for normal cases. This robustness demonstrates its reliability for stage discrimination.

Figure \ref{fig:auroc} illustrates the ROC curves for the Nakagami parameter in hepatic steatosis classification, evaluated on the SMC dataset across three pairwise comparisons: (A) Normal vs. Mild, (B) Mild vs. Severe, and (C) Normal vs. Severe. For comparison (A), the momentum and WMC estimators yielded moderate AUROCs of 0.82–0.83, while the MLE estimator and UNICORN achieved higher values of 0.88 and 0.87, respectively. In comparison (B), UNICORN significantly outperformed baseline methods, exceeding the momentum estimator by 0.12 AUROC points and the MLE estimator by 0.06 AUROC points. This advantage arises from UNICORN’s linear relationship with MRI-PDFF, whereas other methods exhibit saturation at higher fat fractions. For comparison (C), baseline methods demonstrated strong performance (AUROC = 0.90–0.96), but UNICORN attained near-perfect discrimination with an AUROC of 0.99. Collectively, these results confirm UNICORN’s diagnostic superiority, particularly for advanced steatosis stages.

\cref{tab:auroc-1} presents the ROC analysis of UNICORN (Fig.~\ref{fig:auroc}), reporting sensitivity, specificity, positive predictive value (PPV), and negative predictive value (NPV) at three thresholds: (a) the optimal threshold maximizing the F1-score to address class imbalance~\cite{lipton2014optimal}, (b) a fixed 90\% specificity threshold, and (c) a fixed 90\% sensitivity threshold.
For the comparison (A), UNICORN achieves a sensitivity of 0.73 and specificity of 0.87 at the optimal threshold, with corresponding PPV, NPV, and F1-scores of 0.81, 0.80, and 0.77, respectively. At the 90\% specificity threshold, sensitivity decreases to 0.68, while at the 90\% sensitivity threshold, specificity falls to 0.57.
In comparison (B), at the optimal threshold, UNICORN yields a sensitivity of 0.83, specificity of 0.62, PPV of 0.42, NPV of 0.92, and F1-score of 0.55. Under the 90\% specificity constraint, sensitivity drops to 0.38; at 90\% sensitivity, specificity is 0.48.
Notably, for comparison (C), a clinically distinct classification, UNICORN demonstrates exceptional performance, with sensitivity and specificity both exceeding 0.93 across all thresholds. These results underscore UNICORN’s diagnostic robustness, particularly in discriminating extreme-stage steatosis (C), while maintaining adaptability to threshold tuning for nuanced classifications (A, B).

\cref{tab:auroc-2} summarizes the ROC analysis comparing baseline methods with the UNICORN model at the optimal threshold (Fig.~\ref{fig:auroc}). UNICORN achieves the highest AUC for distinguishing severe fatty liver from mild and normal cases. Due to its lower false-positive rate, the optimal threshold for task (A) corresponds to relatively high specificity and low sensitivity. Consequently, UNICORN exhibits substantial gains in specificity (margins: 0.18–0.09) and PPV (margins: 0.06–0.13) for (A), albeit with slight declines in sensitivity and NPV. When evaluated by the F1-score, a metric suited to imbalanced data, UNICORN performs comparably to baselines in (A) and surpasses them in (B) and (C). Notably, for (C), UNICORN excels across all metrics. These results collectively demonstrate UNICORN’s superior ability to discriminate steatosis stages, particularly for advanced cases.

\begin{table}[!ht]
    \caption{ROC curve analysis of UNICORN (\cref{fig:auroc}). (a) optimal threshold, (b) threshold for 90\% specificity, (b) threshold for 90\% sensitivity. PPV denotes positive predictive value and NPV denotes negative predictive value.} 
    \centering
    \scriptsize
    \resizebox{1\linewidth}{!}{
    \begin{tabular}{lcccccc}
          \toprule
          & AUC & Sensitivity & Specificity & PPV & NPV & F1 \\    
           \cmidrule(l){1-1} \cmidrule(l){2-2}  \cmidrule(l){3-3}  \cmidrule(l){4-4} \cmidrule(l){5-5} \cmidrule(l){6-6} \cmidrule(l){7-7} 
           \rowcolor{Gray}
           (A) Normal vs Mild & 0.87 & & & & & \\
           (a) Optimal & & 0.73 & 0.87 & 0.81 & 0.80 & 0.77 \\
           (b) 90\% Specificity & & 0.68 & 0.90 & 0.85 & 0.78 & 0.76 \\
           (c) 90\% Sensitivity & & 0.90 & 0.57 & 0.62 & 0.88 & 0.74 \\
           \cmidrule(l){1-1} \cmidrule(l){2-2} \cmidrule(l){3-7}
           \rowcolor{Gray}
           (B) Mild vs Severe & 0.80 & & & & & \\
           (a) Optimal & & 0.83 & 0.62 & 0.42 & 0.92 & 0.55 \\
           (b) 90\% Specificity & & 0.38 & 0.90 & 0.57 & 0.82 & 0.46 \\
           (c) 90\% Sensitivity & & 0.90 & 0.48 & 0.36 & 0.94 & 0.52 \\
           \cmidrule(l){1-1} \cmidrule(l){2-2} \cmidrule(l){3-7}
           \rowcolor{Gray}
           (C) Normal vs Severe & 0.99 & & & & & \\
           (a) Optimal & & 0.98 & 0.93 & 0.79 & 0.99 & 0.88 \\
           (b) 90\% Specificity & & 0.99 & 0.91 & 0.73 & 1.00 & 0.84 \\
           (c) 90\% Sensitivity & & 0.91 & 0.97 & 0.88 & 0.98 & 0.89 \\
    \bottomrule 
    \end{tabular}
    }
    \label{tab:auroc-1}
\end{table}

\begin{table}[!ht]
    \caption{ROC curve analysis comparing baseline methods corresponding to \cref{fig:auroc}. The threshold is selected by the highest F1 score. PPV denotes positive predictive value, and NPV denotes negative predictive value.} 
    \centering
    \scriptsize
    \resizebox{1\linewidth}{!}{
    \begin{tabular}{clcccccc}
          \toprule
          & Method & AUC & Sensitivity & Specificity & PPV & NPV & F1 \\
           \cmidrule(l){1-1} \cmidrule(l){2-2}  \cmidrule(l){3-3}  \cmidrule(l){4-4} \cmidrule(l){5-5} \cmidrule(l){6-6} \cmidrule(l){7-7} \cmidrule(l){8-8}
           \multirow{4}{*}{\makecell{(A) \\ Normal vs Mild}} & Momentum & 0.83 & 0.85 & 0.68 & 0.68 & 0.85 & 0.76 \\
           & MLE & \bf{0.88} & 0.83 & 0.78 & 0.75 & 0.85 & \bf{0.79}  \\
           & WMC & 0.82 & \bf{0.88} & 0.65 & 0.68 & \bf{0.87} & 0.76 \\
           \cmidrule(l){2-8}
           & UNICORN & 0.87 & 0.73 & \bf{0.87} & \bf{0.81} & 0.80 & 0.77   \\
           \cmidrule(l){1-1} \cmidrule(l){2-8}
           \multirow{4}{*}{\makecell{(B) \\ Mild vs Severe}} & Momentum & 0.68 & 0.70 & 0.60 & 0.36 & 0.86 & 0.47  \\
           & MLE & 0.74 & 0.68 & 0.67 & 0.41 & 0.87 & 0.51  \\
           & WMC & 0.70 & 0.62 & \bf{0.71} & 0.41 & 0.85 & 0.50  \\
           \cmidrule(l){2-8}
           & UNICORN & \bf{0.80} & \bf{0.83} & 0.62 & \bf{0.42} & \bf{0.92} & \bf{0.55}  \\
           \cmidrule(l){1-1} \cmidrule(l){2-8}
           \multirow{4}{*}{\makecell{(C) \\ Normal vs Severe}} & Momentum & 0.90 & 0.76 & 0.89 & 0.64 & 0.93 & 0.69 \\
           & MLE & 0.96 & 0.92 & 0.90 & 0.71 & 0.98 & 0.80 \\
           & WMC & 0.90 & 0.77 & 0.88 & 0.62 & 0.94 & 0.69 \\
           \cmidrule(l){2-8}
           & UNICORN & \bf{0.99} & \bf{0.98} & \bf{0.93} & \bf{0.79} & \bf{0.99} & \bf{0.88}  \\
    \bottomrule 
    \end{tabular}
    }
    \label{tab:auroc-2}
\end{table}

\begin{figure}[!t]
\centering
\includegraphics[width = 0.85\linewidth]{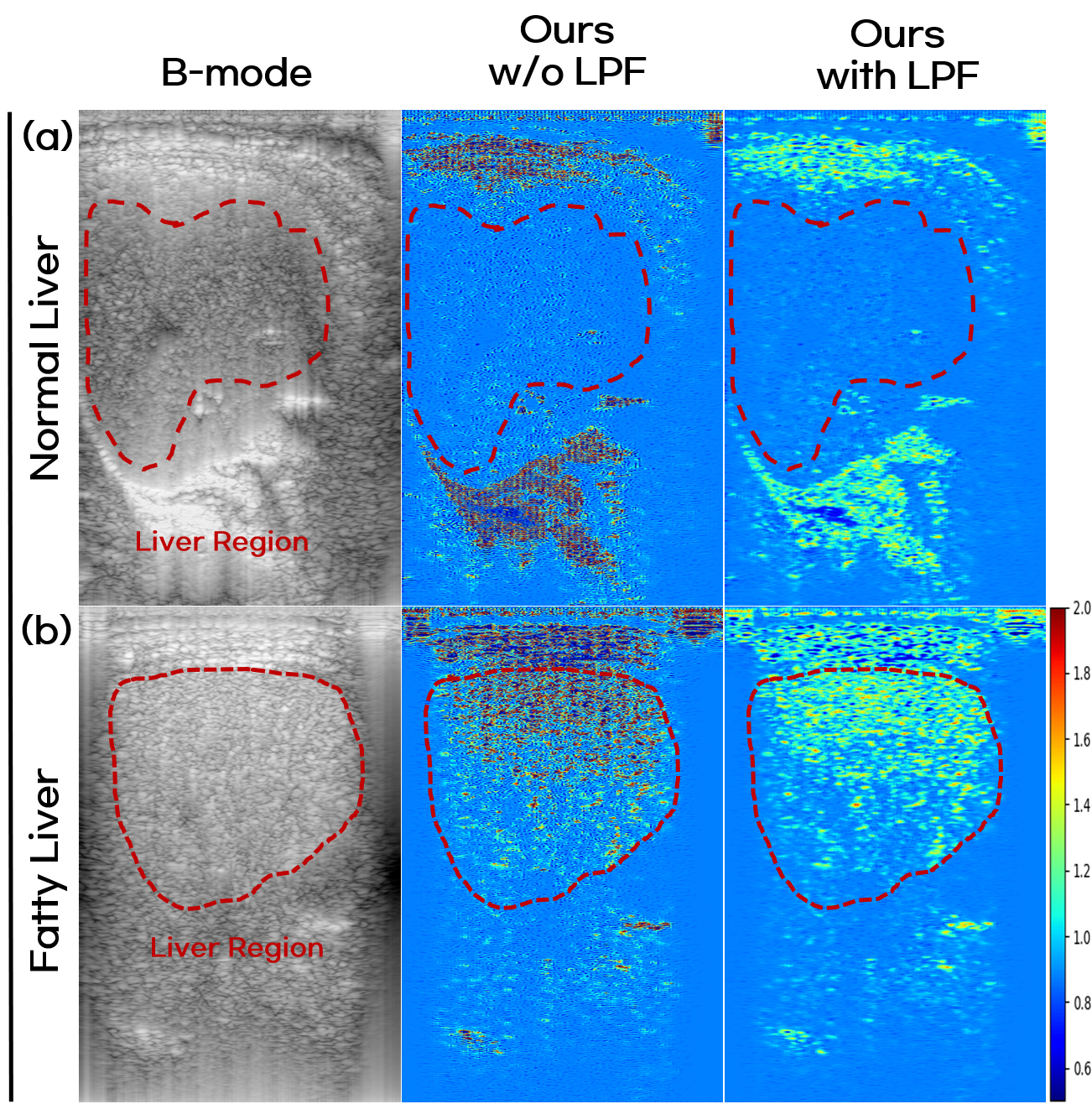}
\caption{Comparison of with and without low-pass filter. (a) represents normal liver cases with MRI-PDFF values of 4.2, and (b) represents fatty liver cases with MRI-PDFF values of 30.4. The red dashed lines indicate the liver region.}
\label{fig:lowpass}
\end{figure}

\section{Ablation Study and Further Analysis}

\subsection{Effect of the Low-pass filter}
To validate the rationale behind using the low-pass filter (LPF), we compared the qualitative results with and without the LPF, as shown in Figure~\ref{fig:lowpass}. Unlike baseline methods that apply an average filter at all sliding window steps, our method applies a single step of the low-pass filter. Therefore, the use of the LPF does not introduce any fairness issues. We found that the results without the LPF show minimal changes, primarily eliminating outliers while still providing a visual distinction between normal and fatty liver. After applying the LPF, we obtained smoother and more visually pleasing estimation results. Therefore, we adopt the low-pass filter in our method.

\begin{table}[!t]
    \caption{Computational complexity including training memory latency and inference speed.} 
    \centering
    \resizebox{1\linewidth}{!}{
    \begin{tabular}{lcccc}
          \toprule
         \cmidrule(lr){2-5} 
           Method  & \bf{Momentum} & \bf{MLE}  & \bf{WMC} & \bf{UNICORN}  \\    
        \cmidrule(l){1-1} \cmidrule(l){2-2}  \cmidrule(l){3-3}  \cmidrule(l){4-4} \cmidrule(l){5-5} 
         Train Memory (GB) & - & - & - & 13.68  \\
         Duration (sec/patient)  &  0.084 $\pm$ 0.003 & 33.79 $\pm$ 2.451  &  0.131 $\pm$ 0.005 &  0.004 $\pm$ 0.001  \\
    \bottomrule 

    \end{tabular}
    }
    \label{time}
\end{table}
\subsection{Comparison on Inference Speed }
To validate the efficiency of UNICORN, we compared the inference speed with other baseline methods, considering training memory latency, as shown in \cref{time}. Our framework does not require additional optimization steps, such as the MLE approach. Once the score model is trained, the inference requires only a single step, accelerated by GPU. The training memory requirement is 13.68 GB, which allows for the use of a single NVIDIA GeForce 3090 GPU. This capability facilitates local implementation even with limited GPU resources. Compared to the existing fastest approach, the momentum-based approach, our method is 20 times faster, making it more suitable for practical applications that require high frames per second (FPS) in ultrasound imaging.

\section{Discussion and Conclusion}

In this work, we introduced UNICORN, a novel framework for ultrasound Nakagami Imaging, addressing limitations of existing methods in visualizing tissue scattering of ultrasound waves. Our method incorporates the score function of a beamformed radiofrequency envelope and its signals, providing a closed-form estimator per pixel followed by low-pass filter adaptation. Unlike conventional methods, we demonstrate that our proposed method calculates the posterior mean, achieving the MMSE from a Bayesian perspective.

Furthermore, traditional methods typically employ a sliding square window to compute the Nakagami parameter, focusing primarily on tissue transitions within the window. This approach, however, does not allow for visualization of the parameter at the individual pixel level. In stark contrast, our pixel-level estimator method is designed to capture the Nakagami parameter directly at each pixel, marking a clear differentiation from prior techniques. 
Accordingly, a key advantage of our approach is its ability to provide both stable and high-resolution ultrasound Nakagami imaging, offering a more nuanced and accurate diagnostic tool compared to conventional methods.

In our simulation experiments, UNICORN demonstrated superior performance over traditional methods, achieving a significant margin of improvement. Furthermore, by applying real ultrasound envelope data, we rigorously validate UNICORN’s effectiveness through in-depth analysis, demonstrating its capability to not only distinguish normal and fatty liver but also differentiate between steatosis stages with high precision. We believe that our framework holds promise for various applications in tumor diagnosis and fat fraction estimation, paving the way for advancements in ultrasound imaging techniques.


\section{Acknowledgments}
The study received approval from the Institutional Review Board (IRB) of Samsung Medical Center (IRB No. SMC 2022-11-026). This work was supported by Samsung Medison Co., Ltd.

\appendix
\section{Proof of \cref{prop}} \label{proof1}
\begin{proof}
For a given Nakagmi distribution ~\cref{pd:naka} and equality ~\cref{key equation}, 
we have 
\begin{align}
   \nabla_r \log p_{R}(r)= (2m-1) \frac{1}{r} -\frac{2m}{\Omega}r 
\end{align}
Accordingly,
\begin{align}
  m\left(\frac{2}{r} - \frac{2r}{\Omega}\right) &= \frac{1}{r} + \nabla_r \log p_{R}(r) 
\end{align}
Furthermore, we have,
\begin{align}
 \hat \Omega = \mathbb{E}[R^2] 
\end{align}
Therefore, we have,
\begin{align}
   \hat m = \mathbb{E}[m|r] &= \frac{\frac{1}{r} + \nabla_r \log p_{R}(r) }{\left(\frac{2}{r} - \frac{2r}{\hat \Omega}\right)}, \quad \text{where} \; \hat \Omega = \mathbb{E}[R^2] 
\end{align}
This concluded the proof.
\end{proof}
\bibliographystyle{splncs04}
\bibliography{reference}

@inproceedings{larrue2011nakagami,
  title={Nakagami imaging with small windows},
  author={Larrue, Aymeric and Noble, J Alison},
  booktitle={2011 IEEE International Symposium on Biomedical Imaging: From Nano to Macro},
  pages={887--890},
  year={2011},
  organization={IEEE}
}

@article{kolar2004estimator,
  title={Estimator comparison of the Nakagami-m parameter and its application in echocardiography},
  author={Kolar, Radim and Jirik, Radovan and Jan, Jiri},
  journal={Radioengineering},
  volume={13},
  number={1},
  pages={8--12},
  year={2004},
  publisher={SPOLECNOST PRO RADIOELEKTRONICKE INZENYRSTVI CZECH TECHNICAL UNIVERSITY~…}
}

@incollection{destrempes2023review,
  title={Review of envelope statistics models for quantitative ultrasound imaging and tissue characterization},
  author={Destrempes, Fran{\c{c}}ois and Cloutier, Guy},
  booktitle={Quantitative ultrasound in soft tissues},
  pages={107--152},
  year={2023},
  publisher={Springer}
}

@article{kawano2013mechanisms,
  title={Mechanisms of hepatic triglyceride accumulation in non-alcoholic fatty liver disease},
  author={Kawano, Yuki and Cohen, David E},
  journal={Journal of gastroenterology},
  volume={48},
  pages={434--441},
  year={2013},
  publisher={Springer}
}

@article{adams2005recent,
  title={Recent concepts in non-alcoholic fatty liver disease},
  author={Adams, LA and Angulo, P},
  journal={Diabetic medicine},
  volume={22},
  number={9},
  pages={1129--1133},
  year={2005},
  publisher={Wiley Online Library}
}

@article{targher2010risk,
  title={Risk of cardiovascular disease in patients with nonalcoholic fatty liver disease},
  author={Targher, Giovanni and Day, Christopher P and Bonora, Enzo},
  journal={New England Journal of Medicine},
  volume={363},
  number={14},
  pages={1341--1350},
  year={2010},
  publisher={Mass Medical Soc}
}

@article{zoller2016nonalcoholic,
  title={Nonalcoholic fatty liver disease and hepatocellular carcinoma},
  author={Zoller, Heinz and Tilg, Herbert},
  journal={Metabolism},
  volume={65},
  number={8},
  pages={1151--1160},
  year={2016},
  publisher={Elsevier}
}

@article{mishra2012epidemiology,
  title={Epidemiology and natural history of non-alcoholic fatty liver disease},
  author={Mishra, Alita and Younossi, Zobair M},
  journal={Journal of clinical and experimental hepatology},
  volume={2},
  number={2},
  pages={135--144},
  year={2012},
  publisher={Elsevier}
}

@article{sumida2014limitations,
  title={Limitations of liver biopsy and non-invasive diagnostic tests for the diagnosis of nonalcoholic fatty liver disease/nonalcoholic steatohepatitis},
  author={Sumida, Yoshio and Nakajima, Atsushi and Itoh, Yoshito},
  journal={World journal of gastroenterology: WJG},
  volume={20},
  number={2},
  pages={475},
  year={2014},
  publisher={Baishideng Publishing Group Inc}
}

@article{ratziu2005sampling,
  title={Sampling variability of liver biopsy in nonalcoholic fatty liver disease},
  author={Ratziu, Vlad and Charlotte, Fr{\'e}d{\'e}ric and Heurtier, Agn{\`e}s and Gombert, Sophie and Giral, Philippe and Bruckert, Eric and Grimaldi, Andr{\'e} and Capron, Fr{\'e}d{\'e}rique and Poynard, Thierry and LIDO Study Group and others},
  journal={Gastroenterology},
  volume={128},
  number={7},
  pages={1898--1906},
  year={2005},
  publisher={Elsevier}
}

@article{kleiner2005design,
  title={Design and validation of a histological scoring system for nonalcoholic fatty liver disease},
  author={Kleiner, David E and Brunt, Elizabeth M and Van Natta, Mark and Behling, Cynthia and Contos, Melissa J and Cummings, Oscar W and Ferrell, Linda D and Liu, Yao-Chang and Torbenson, Michael S and Unalp-Arida, Aynur and others},
  journal={Hepatology},
  volume={41},
  number={6},
  pages={1313--1321},
  year={2005},
  publisher={Wiley Online Library}
}

@article{bravo2001liver,
  title={Liver biopsy},
  author={Bravo, Arturo A and Sheth, Sunil G and Chopra, Sanjiv},
  journal={New England Journal of Medicine},
  volume={344},
  number={7},
  pages={495--500},
  year={2001},
  publisher={Mass Medical Soc}
}

@article{ma2009imaging,
  title={Imaging-based quantification of hepatic fat: methods and clinical applications},
  author={Ma, Xiaozhou and Holalkere, Nagaraj-Setty and Mino-Kenudson, Mari and Hahn, Peter F and Sahani, Dushyant V},
  journal={Radiographics},
  volume={29},
  number={5},
  pages={1253--1277},
  year={2009},
  publisher={Radiological Society of North America}
}

@article{reeder2012proton,
  title={Proton density fat-fraction: a standardized MR-based biomarker of tissue fat concentration},
  author={Reeder, Scott B and Hu, Houchun H and Sirlin, Claude B},
  journal={Journal of magnetic resonance imaging: JMRI},
  volume={36},
  number={5},
  pages={1011},
  year={2012},
  publisher={NIH Public Access}
}

@article{schwenzer2009non,
  title={Non-invasive assessment and quantification of liver steatosis by ultrasound, computed tomography and magnetic resonance},
  author={Schwenzer, Nina F and Springer, Fabian and Schraml, Christina and Stefan, Norbert and Machann, J{\"u}rgen and Schick, Fritz},
  journal={Journal of hepatology},
  volume={51},
  number={3},
  pages={433--445},
  year={2009},
  publisher={Elsevier}
}

@article{tsui2016application,
  title={Application of ultrasound nakagami imaging for the diagnosis of fatty liver},
  author={Tsui, Po-Hsiang and Wan, Yung-Liang},
  journal={Journal of Medical Ultrasound},
  volume={24},
  number={2},
  pages={47--49},
  year={2016},
  publisher={No longer published by Elsevier}
}

@article{tsui2007imaging,
  title={Imaging local scatterer concentrations by the Nakagami statistical model},
  author={Tsui, Po-Hsiang and Chang, Chien-Cheng},
  journal={Ultrasound in medicine \& biology},
  volume={33},
  number={4},
  pages={608--619},
  year={2007},
  publisher={Elsevier}
}

@article{thijssen2008computer,
  title={Computer-aided B-mode ultrasound diagnosis of hepatic steatosis: a feasibility study},
  author={Thijssen, Johan M and Starke, Alexander and Weijers, Gert and Haudum, Alois and Herzog, Kathrin and Wohlsein, Peter and Rehage, Jurgen and De Korte, Chris L},
  journal={ieee transactions on ultrasonics, ferroelectrics, and frequency control},
  volume={55},
  number={6},
  pages={1343--1354},
  year={2008},
  publisher={IEEE}
}

@incollection{nakagami1960m,
  title={The m-distribution—A general formula of intensity distribution of rapid fading},
  author={Nakagami, Minoru},
  booktitle={Statistical methods in radio wave propagation},
  pages={3--36},
  year={1960},
  publisher={Elsevier}
}

@article{dutt1994ultrasound,
  title={Ultrasound echo envelope analysis using a homodyned K distribution signal model},
  author={Dutt, Vinayak and Greenleaf, James F},
  journal={Ultrasonic Imaging},
  volume={16},
  number={4},
  pages={265--287},
  year={1994},
  publisher={Elsevier}
}

@article{chan2021ultrasound,
  title={Ultrasound sample entropy imaging: A new approach for evaluating hepatic steatosis and fibrosis},
  author={Chan, Hsien-Jung and Zhou, Zhuhuang and Fang, Jui and Tai, Dar-In and Tseng, Jeng-Hwei and Lai, Ming-Wei and Hsieh, Bao-Yu and Yamaguchi, Tadashi and Tsui, Po-Hsiang},
  journal={IEEE Journal of Translational Engineering in Health and Medicine},
  volume={9},
  pages={1--12},
  year={2021},
  publisher={IEEE}
}

@article{sato2021fatty,
  title={Fatty liver evaluation with double-Nakagami model under low-resolution conditions},
  author={Sato, Yusuke and Tamura, Kazuki and Mori, Shohei and Tai, Dar-In and Tsui, Po-Hsiang and Yoshida, Kenji and Hirata, Shinnosuke and Maruyama, Hitoshi and Yamaguchi, Tadashi},
  journal={Japanese Journal of Applied Physics},
  volume={60},
  number={SD},
  pages={SDDE06},
  year={2021},
  publisher={IOP Publishing}
}

@article{lin2024clinical,
  title={Clinical performance of ultrasonic backscatter parametric and nonparametric statistics in detecting early hepatic steatosis},
  author={Lin, Chih-Hao and Ho, Ming-Chih and Lee, Po-Chu and Yang, Po-Jen and Jeng, Yung-Ming and Tsai, Jia-Huei and Chen, Chiung-Nien and Chen, Argon},
  journal={Ultrasonics},
  pages={107391},
  year={2024},
  publisher={Elsevier}
}

@inproceedings{lipton2014optimal,
  title={Optimal thresholding of classifiers to maximize F1 measure},
  author={Lipton, Zachary C and Elkan, Charles and Naryanaswamy, Balakrishnan},
  booktitle={Machine Learning and Knowledge Discovery in Databases: European Conference, ECML PKDD 2014, Nancy, France, September 15-19, 2014. Proceedings, Part II 14},
  pages={225--239},
  year={2014},
  organization={Springer}
}

@inproceedings{lim2020ar,
  title={AR-DAE: towards unbiased neural entropy gradient estimation},
  author={Lim, Jae Hyun and Courville, Aaron and Pal, Christopher and Huang, Chin-Wei},
  booktitle={International Conference on Machine Learning},
  pages={6061--6071},
  year={2020},
  organization={PMLR}
}

@article{shukla2019quantitative,
  title={Quantitative imaging biomarkers alliance (QIBA) recommendations for improved precision of DWI and DCE-MRI derived biomarkers in multicenter oncology trials},
  author={Shukla-Dave, Amita and Obuchowski, Nancy A and Chenevert, Thomas L and Jambawalikar, Sachin and Schwartz, Lawrence H and Malyarenko, Dariya and Huang, Wei and Noworolski, Susan M and Young, Robert J and Shiroishi, Mark S and others},
  journal={Journal of Magnetic Resonance Imaging},
  volume={49},
  number={7},
  pages={e101--e121},
  year={2019},
  publisher={Wiley Online Library}
}

@article{fang2020ultrasound,
  title={Ultrasound assessment of hepatic steatosis by using the double Nakagami distribution: a feasibility study},
  author={Fang, Feng and Fang, Jui and Li, Qiang and Tai, Dar-In and Wan, Yung-Liang and Tamura, Kazuki and Yamaguchi, Tadashi and Tsui, Po-Hsiang},
  journal={Diagnostics},
  volume={10},
  number={8},
  pages={557},
  year={2020},
  publisher={MDPI}
}

@article{tsui2010ultrasonic,
  title={Ultrasonic Nakagami imaging: a strategy to visualize the scatterer properties of benign and malignant breast tumors},
  author={Tsui, Po-Hsiang and Yeh, Chih-Kuang and Liao, Yin-Yin and Chang, Chien-Cheng and Kuo, Wen-Hung and Chang, King-Jen and Chen, Chiung-Nien},
  journal={Ultrasound in medicine \& biology},
  volume={36},
  number={2},
  pages={209--217},
  year={2010},
  publisher={Elsevier}
}

@article{ho2012using,
  title={Using ultrasound Nakagami imaging to assess liver fibrosis in rats},
  author={Ho, Ming-Chih and Lin, Jen-Jen and Shu, Yu-Chen and Chen, Chiung-Nien and Chang, King-Jen and Chang, Chien-Cheng and Tsui, Po-Hsiang},
  journal={Ultrasonics},
  volume={52},
  number={2},
  pages={215--222},
  year={2012},
  publisher={Elsevier}
}

@article{rangraz2014nakagami,
  title={Nakagami imaging for detecting thermal lesions induced by high-intensity focused ultrasound in tissue},
  author={Rangraz, Parisa and Behnam, Hamid and Tavakkoli, Jahan},
  journal={Proceedings of the Institution of Mechanical Engineers, Part H: Journal of Engineering in Medicine},
  volume={228},
  number={1},
  pages={19--26},
  year={2014},
  publisher={SAGE Publications Sage UK: London, England}
}

@article{wan2015effects,
  title={Effects of fatty infiltration in human livers on the backscattered statistics of ultrasound imaging},
  author={Wan, Yung-Liang and Tai, Dar-In and Ma, Hsiang-Yang and Chiang, Bing-Hao and Chen, Chin-Kuo and Tsui, Po-Hsiang},
  journal={Proceedings of the Institution of Mechanical Engineers, Part H: Journal of Engineering in Medicine},
  volume={229},
  number={6},
  pages={419--428},
  year={2015},
  publisher={SAGE Publications Sage UK: London, England}
}

@article{shankar2001classification,
  title={Classification of ultrasonic B-mode images of breast masses using Nakagami distribution},
  author={Shankar, P Mohana and Dumane, VA and Reid, John M and Genis, Vladimir and Forsberg, Flemming and Piccoli, Catherine W and Goldberg, Barry B},
  journal={IEEE transactions on ultrasonics, ferroelectrics, and frequency control},
  volume={48},
  number={2},
  pages={569--580},
  year={2001},
  publisher={IEEE}
}

@article{zhang2012feasibility,
  title={Feasibility of using Nakagami distribution in evaluating the formation of ultrasound-induced thermal lesions},
  author={Zhang, Siyuan and Zhou, Fanyu and Wan, Mingxi and Wei, Min and Fu, Quanyou and Wang, Xing and Wang, Supin},
  journal={The Journal of the Acoustical Society of America},
  volume={131},
  number={6},
  pages={4836--4844},
  year={2012},
  publisher={AIP Publishing}
}

@article{tsui2016acoustic,
  title={Acoustic structure quantification by using ultrasound Nakagami imaging for assessing liver fibrosis},
  author={Tsui, Po-Hsiang and Ho, Ming-Chih and Tai, Dar-In and Lin, Ying-Hsiu and Wang, Chiao-Yin and Ma, Hsiang-Yang},
  journal={Scientific reports},
  volume={6},
  number={1},
  pages={33075},
  year={2016},
  publisher={Nature Publishing Group UK London}
}

@article{zhou2018hepatic,
  title={Hepatic steatosis assessment with ultrasound small-window entropy imaging},
  author={Zhou, Zhuhuang and Tai, Dar-In and Wan, Yung-Liang and Tseng, Jeng-Hwei and Lin, Yi-Ru and Wu, Shuicai and Yang, Kuen-Cheh and Liao, Yin-Yin and Yeh, Chih-Kuang and Tsui, Po-Hsiang},
  journal={Ultrasound in medicine \& biology},
  volume={44},
  number={7},
  pages={1327--1340},
  year={2018},
  publisher={Elsevier}
}

@article{zhou2019hepatic,
  title={Hepatic steatosis assessment using quantitative ultrasound parametric imaging based on backscatter envelope statistics},
  author={Zhou, Zhuhuang and Zhang, Qiyu and Wu, Weiwei and Wu, Shuicai and Tsui, Po-Hsiang},
  journal={Applied Sciences},
  volume={9},
  number={4},
  pages={661},
  year={2019},
  publisher={MDPI}
}

@article{park2022quantitative,
  title={Quantitative evaluation of hepatic steatosis using advanced imaging techniques: focusing on new quantitative ultrasound techniques},
  author={Park, Junghoan and Lee, Jeong Min and Lee, Gunwoo and Jeon, Sun Kyung and Joo, Ijin},
  journal={Korean Journal of Radiology},
  volume={23},
  number={1},
  pages={13},
  year={2022},
  publisher={Korean Society of Radiology}
}

@article{zagzebski1993quantitative,
  title={Quantitative ultrasound imaging: in vivo results in normal liver},
  author={Zagzebski, JA and Lu, ZF and Yao, LX},
  journal={Ultrasonic imaging},
  volume={15},
  number={4},
  pages={335--351},
  year={1993},
  publisher={Elsevier}
}

@article{jang2023non,
  title={Non-invasive imaging methods to evaluate non-alcoholic fatty liver disease with fat quantification: a review},
  author={Jang, Weon and Song, Ji Soo},
  journal={Diagnostics},
  volume={13},
  number={11},
  pages={1852},
  year={2023},
  publisher={MDPI}
}

@article{ferraioli2021quantification,
  title={Quantification of liver fat content with ultrasound: a WFUMB position paper},
  author={Ferraioli, Giovanna and Berzigotti, Annalisa and Barr, Richard G and Choi, Byung I and Cui, Xin Wu and Dong, Yi and Gilja, Odd Helge and Lee, Jae Young and Lee, Dong Ho and Moriyasu, Fuminori and others},
  journal={Ultrasound in medicine \& biology},
  volume={47},
  number={10},
  pages={2803--2820},
  year={2021},
  publisher={Elsevier}
}

@article{ghoshal2012ex,
  title={Ex vivo study of quantitative ultrasound parameters in fatty rabbit livers},
  author={Ghoshal, Goutam and Lavarello, Roberto J and Kemmerer, Jeremy P and Miller, Rita J and Oelze, Michael L},
  journal={Ultrasound in medicine \& biology},
  volume={38},
  number={12},
  pages={2238--2248},
  year={2012},
  publisher={Elsevier}
}

@article{rou2024assessment,
  title={Assessment of Hepatic Steatosis Using Ultrasound-Based Techniques: Focus on Fat Quantification},
  author={Rou, Woo Sun and Rou, Woo Sun},
  journal={Clinical Ultrasound},
  volume={9},
  number={1},
  pages={1--17},
  year={2024},
  publisher={The Korean Association of Clinical Ultrasound}
}

@article{li2024correlation,
  title={The correlation between hepatic controlled attenuation parameter (CAP) value and insulin resistance (IR) was stronger than that between body mass index, visceral fat area and IR},
  author={Li, Zhouhuiling and Liu, Renjiao and Gao, Xinying and Hou, Dangmin and Leng, Mingxin and Zhang, Yanju and Du, Meiyang and Zhang, Shi and Li, Chunjun},
  journal={Diabetology \& Metabolic Syndrome},
  volume={16},
  number={1},
  pages={153},
  year={2024},
  publisher={Springer}
}

@article{han2017nakagami,
  title={Nakagami--m parametric imaging for atherosclerotic plaque characterization using the coarse-to-fine method},
  author={Han, Meng and Wan, Jinjin and Zhao, Yongfeng and Zhou, Xiaodong and Wan, Mingxi},
  journal={Ultrasound in Medicine \& Biology},
  volume={43},
  number={6},
  pages={1275--1289},
  year={2017},
  publisher={Elsevier}
}

@article{zhou2018three,
  title={Three-dimensional visualization of ultrasound backscatter statistics by window-modulated compounding Nakagami imaging},
  author={Zhou, Zhuhuang and Wu, Shuicai and Lin, Man-Yen and Fang, Jui and Liu, Hao-Li and Tsui, Po-Hsiang},
  journal={Ultrasonic Imaging},
  volume={40},
  number={3},
  pages={171--189},
  year={2018},
  publisher={SAGE Publications Sage CA: Los Angeles, CA}
}

@article{tsui2014window,
  title={Window-modulated compounding Nakagami imaging for ultrasound tissue characterization},
  author={Tsui, Po-Hsiang and Ma, Hsiang-Yang and Zhou, Zhuhuang and Ho, Ming-Chih and Lee, Yu-Hsin},
  journal={Ultrasonics},
  volume={54},
  number={6},
  pages={1448--1459},
  year={2014},
  publisher={Elsevier}
}

@article{al2020dataset,
  title={Dataset of breast ultrasound images},
  author={Al-Dhabyani, Walid and Gomaa, Mohammed and Khaled, Hussien and Fahmy, Aly},
  journal={Data in brief},
  volume={28},
  pages={104863},
  year={2020},
  publisher={Elsevier}
}

@article{ho2013relationship,
  title={Relationship between ultrasound backscattered statistics and the concentration of fatty droplets in livers: an animal study},
  author={Ho, Ming-Chih and Lee, Yu-Hsin and Jeng, Yung-Ming and Chen, Chiung-Nien and Chang, King-Jen and Tsui, Po-Hsiang},
  journal={PLoS One},
  volume={8},
  number={5},
  pages={e63543},
  year={2013},
  publisher={Public Library of Science San Francisco, USA}
}

@article{lin2015noninvasive,
  title={Noninvasive diagnosis of nonalcoholic fatty liver disease and quantification of liver fat using a new quantitative ultrasound technique},
  author={Lin, Steven C and Heba, Elhamy and Wolfson, Tanya and Ang, Brandon and Gamst, Anthony and Han, Aiguo and Erdman Jr, John W and O’Brien Jr, William D and Andre, Michael P and Sirlin, Claude B and others},
  journal={Clinical Gastroenterology and Hepatology},
  volume={13},
  number={7},
  pages={1337--1345},
  year={2015},
  publisher={Elsevier}
}

@article{insana1990describing,
  title={Describing small-scale structure in random media using pulse-echo ultrasound},
  author={Insana, Michael F and Wagner, Robert F and Brown, David G and Hall, Timothy J},
  journal={The Journal of the Acoustical Society of America},
  volume={87},
  number={1},
  pages={179--192},
  year={1990},
  publisher={Acoustical Society of America}
}

@article{cheng2001maximum,
  title={Maximum-likelihood based estimation of the Nakagami m parameter},
  author={Cheng, Julian and Beaulieu, Norman C},
  journal={IEEE Communications letters},
  volume={5},
  number={3},
  pages={101--103},
  year={2001},
  publisher={IEEE}
}

@article{shankar1995model,
  title={A model for ultrasonic scattering from tissues based on the K distribution},
  author={Shankar, P Mohana},
  journal={Physics in Medicine \& Biology},
  volume={40},
  number={10},
  pages={1633},
  year={1995},
  publisher={IOP Publishing}
}

@article{bamber1981acoustic,
  title={Acoustic properties of normal and cancerous human liver—I. Dependence on pathological condition},
  author={Bamber, JC and Hill, CR},
  journal={Ultrasound in medicine \& biology},
  volume={7},
  number={2},
  pages={121--133},
  year={1981},
  publisher={Elsevier}
}

@article{oelze2016review,
  title={Review of quantitative ultrasound: Envelope statistics and backscatter coefficient imaging and contributions to diagnostic ultrasound},
  author={Oelze, Michael L and Mamou, Jonathan},
  journal={IEEE transactions on ultrasonics, ferroelectrics, and frequency control},
  volume={63},
  number={2},
  pages={336--351},
  year={2016},
  publisher={IEEE}
}

@article{bigelow2005estimation,
  title={Estimation of total attenuation and scatterer size from backscattered ultrasound waveforms},
  author={Bigelow, Timothy A and Oelze, Michael L and O’Brien Jr, William D},
  journal={The Journal of the Acoustical Society of America},
  volume={117},
  number={3},
  pages={1431--1439},
  year={2005},
  publisher={Acoustical Society of America}
}

@article{kim2021noise2score,
  title={Noise2{S}core: {T}weedie’s approach to self-supervised image denoising without clean images},
  author={Kim, Kwanyoung and Ye, Jong Chul},
  journal={Advances in Neural Information Processing Systems},
  volume={34},
  pages={864--874},
  year={2021}
}

@inproceedings{kim2022noise,
  title={Noise distribution adaptive self-supervised image denoising using {T}weedie distribution and score matching},
  author={Kim, Kwanyoung and Kwon, Taesung and Ye, Jong Chul},
  booktitle={Proceedings of the IEEE/CVF Conference on Computer Vision and Pattern Recognition},
  pages={2008--2016},
  year={2022}
}

\end{document}